\documentclass[superscriptaddress,aps,pra,twocolumn,showpacs,nofootinbib,longbibliography,showkeys]{revtex4-2}
\usepackage{slashbox}
\usepackage{amsmath,amssymb,amsthm}
\usepackage{easybmat,comment}
\usepackage[colorlinks=true,citecolor=blue,urlcolor=blue, linkcolor = magenta]{hyperref}
\usepackage[pdftex]{graphicx}
\usepackage{times,txfonts}
\usepackage{braket}
\usepackage{color}
\usepackage{natbib}
\usepackage{subcaption}
\usepackage{ragged2e}
\DeclareCaptionJustification{justified}{\justifying}
 \usepackage[compat=0.4]{yquant}
\newcommand{\be}{\begin{equation}}
	\newcommand{\ee}{\end{equation}}
\newcommand{\ba}{\begin{eqnarray}}
	\newcommand{\ea}{\end{eqnarray}}

\newtheorem{observation}{Observation}

\newtheorem{theorem}{Theorem}

\newtheorem{lemma}[theorem]{Lemma}
\usepackage{comment}
\usepackage{tikz}
\usetikzlibrary{quantikz2}
\begin{document}
	
%\preprint{APS/123-QED}
\title{Bounds on concatenated entanglement-assisted quantum error-correcting codes}
\author{Nihar Ranjan Dash\textsuperscript{}}
   \email{dash.1@iitj.ac.in}
   \affiliation{Indian Institute of Technology Jodhpur, Karwar, Jheepasani, Rajasthan 342030, India}
\author{Sanjoy Dutta}
    \affiliation{Poornaprajna Institute of Scientific Research (PPISR), Bidalur post, Devanahalli, Bengaluru 562164, India}
\author{R. Srikanth\textsuperscript{}}
   \email{srik@ppisr.res.in}
   \affiliation{Poornaprajna Institute of Scientific Research, Bidalur Post, Devanahalli, Bengaluru 562164, India}
\author{Subhashish Banerjee}
   \email{subhashish@iitj.ac.in}
   \affiliation{Indian Institute of Technology Jodhpur, Rajasthan 342030, India}

	%\textsuperscript{b,c}}
%\affiliation{
	%Third institution, the second for Charlie Author
%\affiliation{%


%\date{}% It is always \today, today,

\begin{abstract}
Code concatenation combines two or more component codes to design larger codes with greater noise resilience. Introducing entanglement assistance to concatenated codes provides a further advantage in terms of improved error rates and beating certain bounds on codes that would otherwise be unbeatable. First, we derive the general expression for the shared entanglement of a concatenated code and show that the number of ebits can depend on the order of concatenating the component entanglement-assisted quantum error-correcting codes (EAQECCs). We further construct families of pairs of EAQECCs such that the number of ebits of the resultant of concatenating the two codes in a given pair is order independent. Second, we derive conditions on code distance under which non-maximal-entanglement EAQECCs obtained from a classical quaternary Griesmer or Plotkin code saturate the entanglement-assisted (EA) Griesmer or linear EA Plotkin bound, respectively, extending the known result for maximal-entanglement EAQECCs. Furthermore, we present several families of such nonmaximal-entanglement EAQECCs. Third, we derive an EA version of the quantum Griesmer-Rains bound on the number of correctable errors for EAQECCs. Finally, we present families of pairs of EAQECCs such that the violation of the EA Hamming bound by the resultant of concatenating the two codes in a given pair is order dependent.
 \end{abstract}
%\keywords{Keywords..}
%Use showkeys class option if keyword
\maketitle

%\tableofcontents

\section{Introduction}\label{sec-intro}

Quantum error-correcting codes (QECCs) constitute an important tool for realizing fault-tolerant quantum computation (FTQC), secure quantum communication \cite{preskill2018quantum, deutsch2020harnessing} and also noise characterization \cite{omkar2015characterization}. Subject to orthogonality constraints \cite{calderbank1998quantum}, QECCs can be constructed from classical codes by means of the stabilizer formalism. Entanglement-assisted (EA) quantum error-correcting codes (EAQECCs), which utilize preshared entanglement between the sender and the receiver \cite{brun2006correcting, wilde2008optimal, lai2013dualities, lai2013entanglement, brun2014catalytic} and local encoding implemented by the sender \cite{Kuo2019encoding, wilde2008thesis, wilde2010convolutional, wilde2010convolutionalpra}, provide a generalization of QECCs that overcome the orthogonality constraint and thereby simplify quantum code construction. 

Entanglement-assisted QECCs can enhance the information capacity \cite{bowen2002entanglement, brun2014catalytic} of quantum channels and furthermore violate the nondegenerate quantum EA Hamming bound \cite{li2014entanglement} and the quantum EA Singleton bound \cite{grassl2021entanglementassisted}. The latter violation motivates the need for an EA Singleton bound appropriate for large-distance degenerate codes \cite{grassl2022entropic}, complementing the EA Singleton bound obtained by direct application of the standard Singleton bound to all of the shared qubits \cite{brun2006correcting}. Analogous to the Plotkin bound for classical codes and Griesmer bound for classical and quantum codes, we have for EAQECCs the linear EA Plotkin bound \cite{guo2013linear, lai2013dualities} and EA Griesmer bound \cite{li2015entanglement}. 

It is appropriate to quantify the net key rate of an EAQECC as the number of logical qubits transmitted minus that of preshared ebits. Thus it is practically advantageous to keep the shared entanglement low \cite{hsieh2009entanglement}, leading to the concept of catalytic EAQECCs \cite{brun2014catalytic}, which is relevant for FTQC. In the other extreme, the number of ebits is maximal when it equals the code length of the given EAQECC minus the logical key rate \cite{lai2013dualities}. From an experimental perspective, an all-optical implementation of EAQECCs has been proposed \cite{djordjevic2010photonic}, including the photonic encoder and decoder architectures required for the same.

Entanglement-assisted QECCs have been generalized in various directions; for example, a continuous-variable EAQECC and its practical implementation using linear optics were proposed in \cite{wilde2007optics} and certain works have studied relaxing the original assumption of noiselessness of the ebits \cite{lai2012entanglement-assisted, wang2014quantum}. Moreover, shorter EAQECCs can be combined to produce longer concatenated EAQECCs (CEAQECCs) \cite{fan2022entanglementassisted}. In classical coding, concatenated codes were introduced by Forney to achieve large coding yields with less encoding as well as decoding complexity \cite{forney1966concatenated}. This was subsequently extended to conventional concatenated (stabilizer) codes in the quantum case by various authors \cite{gottesman1997stabilizer, knill1996concatenated, gaitan2008quantum}. A classical generalization of Forney's scheme was proposed by Blokh and Zyablov \cite{blokh1974coding} using nested inner subcodes and multiple outer codes, later reformulated by Zinov'ev \cite{zinov1976generalized} as a hierarchical cascade, i.e., the generalized concatenated code framework, subsequently extended to the quantum domain by Grassl \textit{et al.} \cite{grassl2009generalized}. Concatenated EAQECCs bring together the benefits of shared entanglement and concatenation leading to codes that outperform traditional QECCs and EAQECCs \cite{fan2022entanglementassisted}.
Entanglement-assisted QECCs have been combined with passive codes \cite{lidar2001decoherenceI, lidar2001decoherenceII, dash2023concatenating}, leading to subsystem EAQECCs \cite{hsieh2007general}. 

The phenomenon of degeneracy of errors is a distinguishing feature of QECCs absent in classical codes and as such inherited by EAQECCs \cite{li2014entanglement}. It arises naturally in the context of noiseless subspaces \cite{dash2023concatenating} as well as code concatenation of possibly nondegenerate component codes \cite{fan2021degenerate_conc}. In the latter case, the degeneracy arises as the distance of the concatenated code is lower bounded by the product of those of the constituent codes. 

In this article we investigate the different bounds arising for CEAQECCs and how these bounds relate to the properties of the component EAQECCs. \color{black} After presenting preliminaries (Sec. \ref{sec-prelim}), we begin by investigating the saturation of various bounds by EAQECCs derived from classical quaternary codes or from standard stabilizer codes (Sec. \ref{sec-saturation_bounds_eaqeccs}). Specifically, we derive general conditions on code distance such that nonmaximal-entanglement EAQECCs, obtained from a classical quaternary Griesmer or Plotkin code, saturate the EA Griesmer or linear EA Plotkin bound, respectively. Further, for EAQECCs $[[n,k,d;c]]_{q}$ ($d\ge q^2$) we derive an EA Griesmer-Rains bound on the number of correctable errors for EAQECCs. 

These results form the basis for the subsequent sections, where we explore how these properties are propagated under concatenation of EAQECCs. However, prior to that (taken up in Sec. \ref{sec-violation_or_saturation_ceaqeccs}) we need to determine the implications for the resultant concatenated code arising from whether or not the inner code rate divides the outer code length (addressed in Sec. \ref{sec-eacqeccs}) and then to determine how the order of concatenation affects the performance of CEAQECCs in terms of parameters such as the number of shared ebits, the logical error probability and the pseudothreshold (addressed in Sec.  \ref{sec-optimizing_order_less_ebits}). \color{black} Specifically, in the latter section, we identify families of pairs of EAQECCs such that concatenating the two codes in a given pair entails the same number of ebits, independently of the order of concatenation, e.g., the CEAQECCs obtained by combining an $[[n+1,k;c+1]]$ code with an $[[n,k-1;c+1]]$ code or vice versa, and likewise the CEAQECCs obtained by concatenating any two members of the family of $[[n,1;n-1]]$ EAQECCs ($n \in \mathbb{Z}^+, Z\ge3$).  

\color{black} In Sec. \ref{sec-violation_or_saturation_ceaqeccs}, we consider certain aspects of propagation of various bound-saturation or bound-violation properties in CEAQECCs. For example, \textcolor{black}{we show that the EA Hamming bound (EAHB) satisfying codes $[[n, 1, n; n-1]]$ (odd $n$) and $[[n, 1, n-1; n-1]]$ (even $n$), concatenated with the EAHB-violating (degenerate) $[[8, 1, 5; 1]]$ code, will result in CEAQECCs that violate the EAHB for sufficiently large $n$ if and only if the former serves as the outer code (Sec. \ref{sec-violation_or_saturation_ceaqeccs}).}  We indicate certain families of EAQECCs for which the saturation of various bounds is closed under concatenation. \color{black}
Finally, we present our conclusions and related discussion in Sec. \ref{sec-conclusion}, indicating a number of open problems emerging from our work.

\section{Preliminaries}\label{sec-prelim}
\subsection{QECCs}
Let $\mathcal{S} \in \mathcal{P}_{n}$ be an Abelian subgroup and $-I \not\in \mathcal{S}$, with $n-k$ independent generators \cite{lidarandbrun2013qecbook}. The corresponding code $\mathcal{C(S)}$ of dimension $2^k$ is the subspace fixed by the elements of $\mathcal{S}$. The code $\mathcal{C(S)}$ constitutes an $[[n,k,d]]$ quantum stabilizer code \cite{gottesman1997stabilizer} which encodes $k$ logical qubits into $n$ physical qubits with code distance $d$. A set of errors $\{E_i\}$ is correctable by the code $\mathcal{C(S)}$ if and only if $E_i^{\dagger} E_j \not\in \mathcal{N(S)} - \mathcal{S}$ ($i\ne j$), where $\mathcal{N(S)}$ is the normalizer of the stabilizer $\mathcal{S}$ \cite{lidarandbrun2013qecbook, nielsen2010quantum}. Any element in $\mathcal{N(S)} - \mathcal{S}$ has a minimum weight equal to the minimum distance of code $\mathcal{C(S)}$.  A code is degenerate if $E_iE_j \in \mathcal{S}$ or equivalently if there is any element in $\mathcal{S}$ with weight less than the code distance $d$. 

\subsection{EAQECCs}
Entanglement-assisted QECCs are QECCs that take advantage of the preshared entanglement \cite{brun2006correcting, brun2014catalytic}. Other frameworks that introduce entanglement assistance for error correction include entanglement-assisted operator quantum error correction \cite{hsieh2007general} and entanglement-assisted classical enhancement of quantum error correction \cite{kremsky2008classical}. A unified approach to these frameworks was recently proposed in Ref. \cite{nadkarni2024unified}. Somewhat different ideas, namely entanglement-assisted implementation of fault-tolerant encoder and decoder architectures for qudit stabilizer codes, were discussed in Ref. \cite{sharma2024fault}.

Here our focus is on the EAQECC framework. An EAQECC is represented by a $[[n,k,d;c]]$ code, which encodes $k$ logical qubits into $n$ physical qubits prepared using $c$ qubits that are halves of as many ebits. Here $d$ is the minimum distance of the code. The logical state is given by \cite{lai2012entanglement-assisted}
\begin{align}
    \ket{\psi_{L}}=(U_{A} \otimes I_{B})\big[\ket{\psi} \otimes \ket{0}^{\otimes (n-k-c)} \otimes (\ket{\phi^{+}}_{AB}^{\otimes c})\big],
    \label{eq: logicalstateeaqecc}
\end{align}
where $\ket{\phi^{+}}_{AB}$ represents the maximally entangled state preshared between the sender Alice and the receiver Bob. Alice uses a unitary operator $U$ to encode the state $\ket{\psi}$ along with the ancilla and her part of the ebits. A maximal-entanglement EAQECC is one for which $c=n-k$ in Eq. (\ref{eq: logicalstateeaqecc}), and hence no ancillas are required  \cite{lai2013dualities}. The schematic diagram of EAQECC is depicted in Fig. \ref{fig:schematiceaqecc}.
\begin{figure}[t]
    \centering
    \includegraphics[width=8.5cm]{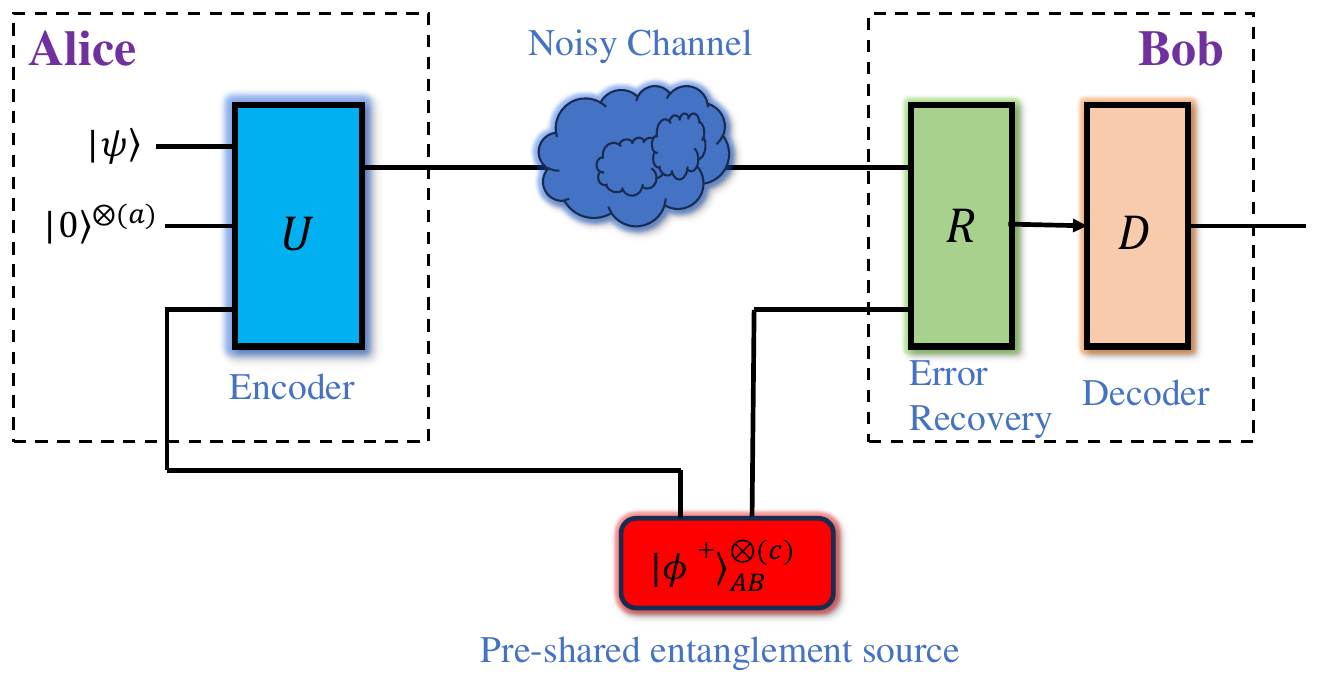}
    \caption{Schematic diagram of the operation of an EAQECC.  Entanglement is pre-shared between Alice and Bob. Alice encodes the $k$-qubit state $\ket{\psi}$ using the $a \equiv n-k-c$ ancilla and her part of ebits, with the help of the encoder $U$ [Eq. (\ref{eq: logicalstateeaqecc})], and transmits the $n$ qubits via a noisy quantum channel. Bob measures the qubits that he received over the channel, together with his half of the ebits to determine the error syndromes, by which he recovers $\ket{\psi}$.}
    \label{fig:schematiceaqecc}
\end{figure} 
 
Let $\mathcal{S}^{\prime} \in \mathcal{P}_{n}$ be a non-Abelian subgroup with $2^{2c+a}$ distinct elements up to an overall phase. Then there exists a set of generators $\{\Bar{Z}_1,\dots,\Bar{Z}_{c+a},\Bar{X}_{1},\dots,\Bar{X}_{c}\}$ for $\mathcal{S}^{\prime}$, satisfying the following commutation relations \cite{lidarandbrun2013qecbook, lai2013entanglement}:
\begin{subequations}
    \begin{align}
        [\Bar{Z}_i,\Bar{Z}_j]&=0  \forall i,j,\\
        [\Bar{X}_i,\Bar{X}_j]&=0  \forall i,j,\\
        [\Bar{X}_i,\Bar{Z}_j]&=0  \forall i\neq j,\\
        \{\Bar{X}_i,\Bar{Z}_i\}&=0  \forall i=j.
    \end{align}
\end{subequations}
Here $\mathcal{S}_{S}^{\prime}=\langle \Bar{Z}_1,\dots,\Bar{Z}_{c},\Bar{X}_{1},\dots,\Bar{X}_{c} \rangle$ denotes the symplectic subgroup and $\mathcal{S}_{I}^{\prime}=\langle \Bar{Z}_{c+1},\dots,\Bar{Z}_{c+a} \rangle$ the isotropic subgroup. Thus $\mathcal{S}^{\prime}$ is generated by $\mathcal{S}_{I}^{'}$ and $\mathcal{S}_{S}^{\prime}$. The number of anticommuting pairs in $\mathcal{S}_{S}^{\prime}$, namely, $c$, is the number of ebits, while the number of (commuting) elements in $\mathcal{S}_{I}^{\prime}$ gives the number of ancillas, namely $a$. 

A set of errors $\{E_i\}$ is correctable by the code $\mathcal{C}(\mathcal{S}^{\prime})$ if and only if $E_i^{\dagger} E_j \not\in \mathcal{N}(\mathcal{S}^{\prime}) - \mathcal{S}_{I}^{\prime}$, where $\mathcal{N}(\mathcal{S}^{\prime})$ is the normalizer of the group $\mathcal{S}^{\prime}$. Any element in $\mathcal{N}(\mathcal{S}^{\prime}) - \mathcal{S}_{I}^{\prime}$ has a minimum weight equal to the minimum distance of code $\mathcal{C}(\mathcal{S}^{\prime})$ \cite{brun2006correcting}.
The code rate $r$ and the net rate $r_n$ for an $[[n,k,d;c]]$ code are given by $r=\frac{k}{n}$ and $r_n=\frac{k-c}{n}$, respectively. We can consider $r_e \equiv \frac{c}{n}$ as the entanglement-assistance rate. 

\subsection{Bounds for EAQECCs}

A nondegenerate $[[n,k,d;c]]$ code satisfies the EA Singleton bound \cite{brun2006correcting}
\begin{equation}
       k\leq c+n-2d+2.
       \label{eq:singleton}
\end{equation}
In the degenerate case, Eq. (\ref{eq:singleton}) holds only if the no-cloning bound of $d\leq\frac{n}{2}+1$ holds
\cite{lai2018linear}; otherwise, i.e., if $d\ge\frac{n}{2}+1$, the EA Singleton bound is given by \cite{grassl2022entropic}
    \begin{align}
    k\leq \frac{(n-d+1)(c+2d-2-n)}{3d-3-n}.
    \label{eq:qsb_d_geq}
\end{align}
Note that setting $d=\frac{n}{2}+1$, both Eqs. (\ref{eq:singleton}) and (\ref{eq:qsb_d_geq}) imply $k\le c$.
It is worth noting that code parameters by themselves do not necessarily determine whether the code is degenerate \cite{lai2013entanglement}. 
%\textcolor{green}{\textbf{For example, EAQECCs such as $[[7,1,5; 3]]$, $[[7,1,5; 4]]$, $[[7,1,5; 5]]$, and $[[9,1,7;7]]$ can be either degenerate or non-degenerate, with explicit constructions given in \cite{lai2013entanglement} by optimization over $[[7,1,3]]$ quantum BCH (Bose-Chaudhuri-Hocquenghem) code in the case of the first three, and Shor's nine-qubit code, in the last. Notice that these codes are with $d>\frac{n}{2}+1$.}} 
Examples of degenerate EAQECCs \cite{li2015entanglement} with $d<\frac{n}{2}+1$ and $d>\frac{n}{2}+1$ are $[[10,2,5;2]]$ and $[[19,2,12;11]]$, respectively. The EAQECCs saturating the bounds in Eq. (\ref{eq:singleton}) or (\ref{eq:qsb_d_geq}) are EA maximum distance separable (MDS) codes. An MDS code saturates the Singleton bound and thus offers the maximum possible distance for a given code length, rate, and alphabet size.  All our results in Sec. \ref{sec-saturation_bounds_eaqeccs} obey the EA Singleton bound (\ref{eq:singleton}), whereas the EA Singleton bound (\ref{eq:qsb_d_geq}) is applicable for some of our CEAQECCs in Sec. \ref{sec-violation_or_saturation_ceaqeccs}. 

For a nondegenerate $[[n,k,d;c]]$ code the EAHB \cite{lai2018linear, fan2022entanglementassisted}
\begin{align}
     \sum_{i=0}^{t} 3^{i}\binom{n}{i} \leq 2^{n-k+c},
     \label{eq:eahb}
\end{align}
where the code can correct $t=\bigl \lfloor\frac{d-1}{2}\bigr\rfloor$ errors. Based on Eq. (\ref{eq:eahb}), the EA-Hamming efficiency \cite{dash2023concatenating} can be defined as
\begin{align}
    \varphi \equiv
    \left(\frac{\log_2\left[\sum_{i=0}^{t} 3^{i}\binom{n}{i}\right]}{n-k+c}\right).
    \label{eq:hammingefficiency}
\end{align}
As degenerate EAQECCs and CEAQECCs may violate the EAHB, we use $\varphi$ as a performance metric to quantify the degree of this violation.

A tighter version of the Singleton bound, taking into consideration the dimension $q$ of the alphabet, is provided by the Griesmer bound \cite{Griesmer1960, SOLOMON1965170, macwilliams1978theory, huffman2003fundamentals}, which is suitable to judge the optimality of QECCs constructed from classical codes \cite{li2015entanglement}. For a classical $[n,k,d]_{q}$ code with $k\ge 1$, the Griesmer bound is given by \cite{macwilliams1978theory, li2015entanglement}
\begin{align}
n \geq \sum_{i=0}^{k-1}\biggl \lceil \frac{d}{q^{i}} \biggr \rceil,
 \label{eq:classicalgriesmerbound}
\end{align}
where $\lceil \cdot \rceil$ is the ceiling function. We refer to a code saturating the Griesmer bound as a Griesmer code.

\begin{comment}
\textcolor{green}{\textbf{The Griesmer bound for $[[n,k,d]]_{q}$ linear QECCs is derived from that for the classical $[\frac{n+k}{2},k]_{q^2}$ as the minimum distance of the quantum $q$-ary code is at most the minimum distance of the classical $q^{2}$-ary code \cite{Ashikhmin1999upper, sarvepalli2010}. Thus for $[[n,k,d]]_{q}$ linear QECCs, the Griesmer bound is given by \cite{li2015entanglement}
\begin{align}
    \frac{n+k}{2} \geq \sum_{i=0}^{k-1}\Biggl \lceil \frac{d}{{q^{2i}}} \Biggr \rceil,
    \label{eq:quantumgriesmer}
    % Eq. (8)
\end{align}}}
\end{comment}
For any $[[n,k,d;c]]_{q}$ code derived from classical $q^{2}$-ary code, the EA Griesmer bound is \cite{li2015entanglement,luo2023constructing}
\begin{align}
    \frac{n+c+k}{2} \geq \sum_{i=0}^{k-1}\Biggl \lceil \frac{d}{q^{2i}} \Biggr \rceil,
    \label{eq:eagriesmer}
\end{align}
which generalizes the corresponding non-EA result \cite{Ashikhmin1999upper, sarvepalli2010}. For the case of the maximal-entanglement binary EAQECC, setting $c {:=} n-k$ in Eq. (\ref{eq:eagriesmer}), we obtain the classical bound (\ref{eq:classicalgriesmerbound}) for the quaternary ($q=4$) code. By expanding the summation in Eq. (\ref{eq:eagriesmer}), we find that the right-hand side (rhs) here is greater than or equal to $d + k-1$. Referring to Eq. (\ref{eq:singleton}), it follows that the EA Griesmer bound implies the EA Singleton bound, paralleling a similar result for the quantum Griesmer for Calderbank-Shor-Steane (CSS) codes and quantum Singleton bounds \cite{sarvepalli2010}. However, saturation of the EA Griesmer bound does not imply EA MDS codes. An example here would be $[[12,2,6;2]]$ \cite{li2015entanglement}. A code is called optimal in the Griesmer sense if $d$ attains the maximum allowed by the bound, for the given parameters $n$ and $k$.

The Plotkin bound \cite{Plotkin1960} is tighter than the Singleton bound \cite{guo2013linear} and sometimes also tighter than the Hamming bound for large $d$, even approaching $n$ \cite{huffman2003fundamentals}. For any classical $[n,k,d]_q$ code such that $d >(1-\frac{1}{q})n$, the Plotkin bound is given by \cite{huffman2003fundamentals, macwilliams1978theory} 
    \begin{align}
    \frac{(q-1)n\times q^{k}}{q(q^{k}-1)} \ge d.
    \label{eq:classicalplotkinboundqary}
\end{align} 
We refer to a code saturating the Plotkin bound as a Plotkin code.
%\textcolor{green}{\textbf{Thus, the Plotkin bound for a classical $[n,k,d]_4$ code is \begin{align}
%    \frac{3n\times 4^{k}}{4(4^{k}-1)} \ge d.
%    \label{eq:classicalplotkinbound}
    % Eq 11
%\end{align}}}
\begin{comment}
\textcolor{green}{\textbf{Ref. \cite{guo2013linear} provides the linear EA-Plotkin bound for linear EAQECC $[[n,k,d;c]]$ derived from a classical quarternary code:
\begin{align}
    \frac{3\times 4^{k}}{8(4^{k}-1)} \left(n+c+k\right) \geq d.
    % Eq (12)
    \label{eq:plotkinboundlinear}
\end{align}}}
\end{comment}

The general linear EA Plotkin bound for linear EAQECC $[[n,k,d;c]]_{q}$ derived from a classical $q^{2}$-ary code can be generalized as \cite{guo2013linear}
\begin{align}
    \frac{(q^{2}-1)\times q^{2k}}{2q^{2}(q^{2k}-1)} \left(n+c+k\right) \geq d.
    \label{eq:plotkinboundlinearqary}
\end{align}
Here we obtain the bound for the maximal-entanglement EAQECC \cite{lai2013dualities}, setting $c {:=} n-k$ in Eq. (\ref{eq:plotkinboundlinearqary}).

If we remove the ceiling function from the right-hand side of the EA Griesmer bound, the result is a geometric progression for finite terms, which gives $d\left(\frac{1-q^{-2k}}{1-q^{-2}}\right) = d\frac{q^2}{q^2-1}\frac{q^{2k}-1}{q^{2k}}$. Inserting this in Eq. (\ref{eq:eagriesmer}), we obtain Eq. (\ref{eq:plotkinboundlinearqary}). Thus, the EA-Griesmer bound implies the linear EA-Plotkin bound. The EA-Singleton and linear EA-Plotkin bounds are in general incomparable and both are weaker than the EA-Griesmer bound. However, for $k=1$, EA Griesmer bound matches the EA Singleton bound Eq. (\ref{eq:singleton}) and linear EA Plotkin bounds. Thus, the codes violating the EA Singleton bound Eq. (\ref{eq:singleton}) presented in \cite{grassl2021entanglementassisted} also violate the linear EA Plotkin and EA Griesmer bounds. For example the optimal $[[9,1,6;1]]$ code \cite{grassl2021entanglementassisted} with $d \geq \frac{n}{2}+1$ violates this triad. Note that for $k \geq 2$, it is stronger than both. Some of the EAQECCs given in \cite{li2015entanglement} saturate the EA Griesmer bound, while Ref. \cite{guo2013linear} constructs EAQECCs that saturate the linear EA Plotkin bound. It is natural to expect that CEAQECCs obtained from concatenating components characterized $k=1$ should also meet the triad of EA Singleton, EA Griesmer and linear EA Plotkin bounds simultaneously. We show that this is indeed the case in Sec. \ref{sec-saturation_bounds_eaqeccs}.

\subsection{Generalized concatenated quantum codes}
Historically, the idea of concatenated codes dates back to Forney \cite{forney1966concatenated}, whose scheme formed the basis of conventional concatenated (stabilizer) codes in the quantum case \cite{gottesman1997stabilizer, knill1996concatenated, gaitan2008quantum}. The direct quantum counterpart of the Forney code was introduced in Refs. \cite{grassl2009generalized, gottesman2024surviving}, where an \(((n_o,K_o,d_o))_q\) outer code over a \(q\)-ary alphabet was concatenated with an \(((n_i,q,d_i))_r\) inner code by replacing each outer symbol by one of \(q\) inner codewords, yielding an overall code of parameters \(((n_on_i, K_o))_r\). Here we use the notation (applicable to stabilizer and nonadditive codes \cite{sarvepalli2010}) whereby $((n, K))_{q}$ is equivalent to $[[n, \log_q K]]_q$. Blokh and Zyablov \cite{blokh1974coding} generalized classical Forney concatenation to multilevel concatenation by partitioning inner symbols in stages and applying multiple outer codes; Zinov'ev \cite{zinov1976generalized} introduced a so-called cascade perspective, emphasizing partitions of the inner symbol set handled by an outer code.

Extending the above ideas to the quantum domain, Grassl \textit{et al.} \cite{grassl2009generalized} established the generalized concatenated quantum code (GCQC) framework. In their construction, the Hilbert space of $n_i$ qudits for each inner block is organized via a partition of the computational basis into $q$ sets, each of which, after stabilizer projection, defines a subcode $\mathcal{Q}_j$ with parameters $((n_i,\varsigma,\delta))_r$, where $\varsigma$ and $\delta$ denote the protected dimension and distance of the subcode, respectively. The outer \(q\)-ary code determines which subcode is used for each block, and logical information is encoded with that subcode. The composite quantum code has the parameter form
\(
\bigl(\bigl(n_on_i,\; K_o  \varsigma^{n_o},\; \Delta \bigr)\bigr)_r,
\)
for some distance $\Delta$, and is typically nonadditive. This construction reduces to the Forney concatenation in the special case \(\varsigma=1\) (where each outer code symbol maps to an inner code vector rather than subspace) yet retains the multilevel partitioning spirit inherited from Blokh and Zyablov and from Zinov'ev. The assumption that each subcode $\mathcal{Q}_j$ is of identical dimension $\varsigma$ is not essential to the GCQC construction but provides a simplification that allows derivation of properties such as minimum distance of the code \cite{grassl2009generalized} and the reduction to the Forney code. In light of the GCQC approach, we find that there can be more than one way to realize concatenated codes, and our present work focuses on a particular approach, namely the framework of conventional concatenated quantum codes \cite{gottesman1997stabilizer, gaitan2008quantum}. In Sec. \ref{sec-eacqeccs} we point out that Forney codes can be encompassed in the conventional framework, but the generalization due to Grassl \textit{et al.} \cite{grassl2009generalized} goes beyond it. 

%\textcolor{blue}{It is worth noting that generalized concatenated codes (GCCs), as introduced by Grassl \emph{et al.} \cite{grassl2009generalized}, do not encompass all forms of concatenation used in quantum error correction. In GCC, each outer code symbol is associated with a single block of the inner code, which is partitioned into $q$ orthogonal subspaces of equal dimension $\delta$. The outer symbol selects the subspace, while quantum information is carried within it. In the special case $\delta=1$, this reduces to Forney's concatenation scheme where each outer symbol maps to a single inner codeword. By contrast, conventional concatenated quantum codes allow inner codes with rate $k_i>1$ \cite{gottesman1997stabilizer, gaitan2008quantum}, so that one inner block corresponds to a \emph{block of outer symbols}. Such blockwise concatenation lies outside the strict GCC framework, making GCC and conventional concatenation related but distinct paradigms.}
\color{black}

\section{EAQECCs saturating EA Singleton, EA Griesmer and linear EA Plotkin bounds}\label{sec-saturation_bounds_eaqeccs}
As pointed out in \cite{lai2012entanglement-assisted}, it is possible to convert any standard $[[n,k,d]]$ stabilizer code into an $[[n-c,k,d;c]]$ EAQECC. Moreover, if the standard stabilizer code saturates the quantum Singleton bound, then its derived EAQECCs also saturate the EA Singleton bound (\ref{eq:singleton}). Here we show that this property holds true also for the linear EA Plotkin and EA Griesmer bounds.

\begin{theorem}
    If a standard $[[n,k,d]]$ code saturates the Griesmer bound (linear Plotkin bound), then the derived $[[n-c,k,d;c]]$ EAQECC ($0 \le c \le n-k$) saturates the EA Griesmer bound (linear EA Plotkin bound).
    \label{th:griemserandplotkin}
\end{theorem}
\begin{proof}
The EA Griesmer bound and linear EA Plotkin bound for an $[[n,k,d;c]]$ code are given by Eqs. (\ref{eq:eagriesmer}) and (\ref{eq:plotkinboundlinearqary}), respectively (setting $q=2$). The dependence of the saturation condition on $n$ and $c$ appears through their sum $n+c$, which is invariant when a QECC is transformed to its derived EAQECC.
\end{proof}
As an illustration of Theorem \ref{th:griemserandplotkin}, the EA Griesmer bound and the linear EA Plotkin bound are both saturated by the Bowen-type $[[3,1,3;2]]_{B}$ \cite{bowen2002entanglement} and $[[4,1,3;1]]_{B}$ \cite{lai2012entanglement-assisted} codes derived from the $[[5,1,3]]$ code. Note that codes that are not derivatives, but share the same parameters $(n,k,d,c)$, will also saturate these bounds.

\begin{theorem}
    Given an $[[n,1,d;c]]$ EAQECC that is non-degenerate or degenerate with $d\leq \frac{n}{2}+1$, such a code saturates either the EA Singleton, EA Griesmer, and linear EA Plotkin bounds or none of the three.
    \label{thm:3in1}
\end{theorem}
\begin{proof}
    For an $[[n,1,d;c]]$ EAQECC, we find from Eqs.  (\ref{eq:singleton}), (\ref{eq:eagriesmer}) and (\ref{eq:plotkinboundlinearqary}) (setting $q=2$) that 
 the condition for saturation of all these three bounds is identically $d=\frac{n+1+c}{2}$.
\end{proof}

From a classical $[n,k,d]$ quaternary code, one can construct a nondegenerate $[[n,2k-n+c,d;c]]$ EAQECC \cite{brun2006correcting, brun2014catalytic}, and we know that if that classical quaternary code saturates the Singleton bound, then its corresponding EAQECC also saturates the EA-Singleton bound \cite{brun2006correcting}. An analogous result exists for the  EA Griesmer \cite{li2015entanglement} and EA Plotkin \cite{lu2015maximal} bounds in the case of maximal-entanglement EAQECCs. Here we indicate more general results, covering nonmaximal-entanglement codes that saturate the EA Griesmer and linear EA Plotkin bounds.
\begin{theorem}
The $[[n,\kappa \equiv 2k-n+c,d;c]]$ EAQECCs induced by a classical $[n,k,d]_{4}$ Griesmer code will satisfy the EA Griesmer bound and saturate it if and only if the EAQECC is a maximal-entanglement code or $d \le 4^{\kappa}$.
\label{thm:featuregriesmer}
\end{theorem}
\begin{proof}
By assumption, the classical $[n,k,d]_{4}$ Griesmer code saturates the classical Griesmer bound (\ref{eq:classicalgriesmerbound})
\begin{align}
    n = \sum_{i=0}^{k-1}\biggl \lceil \frac{d}{4^{i}} \biggr \rceil \equiv \sum_{i=0}^{k-a/2-1}\biggl \lceil \frac{d}{4^{i}} \biggr \rceil + \sum_{i=k-a/2}^{k-1}\biggl \lceil \frac{d}{4^{i}} \biggr \rceil, ~~ a\ge2.
    \label{eq:classicalgriesmerbound=}
\end{align}
We know that this classical code induces $[[n,\kappa,d;c]]$ EAQECCs. Noting that $c+\kappa\equiv n-a$, it follows that $\frac{a}{2} = n-k-c$. We can then rewrite Eq. (\ref{eq:eagriesmer}) for the EA Griesmer bound as 
\begin{align}
 n \geq \sum_{i=0}^{k-a/2-1}\biggl \lceil \frac{d}{4^{i}} \biggr \rceil + \frac{a}{2}.
 \label{eq:eagriesmercheck}
\end{align}
For $a=0$, the saturation clearly holds. Now consider the case of $a>0$. Note that the second summand on the right-hand side of Eq. (\ref{eq:classicalgriesmerbound=}) has exactly $\frac{a}{2}$ terms. Comparing the right-hand sides of Eqs. (\ref{eq:eagriesmercheck}) and (\ref{eq:classicalgriesmerbound=}), we find that saturation holds if and only if the second summand in Eq. (\ref{eq:classicalgriesmerbound=}) evaluates to $\frac{a}{2}$, i.e., each term in the summation is 1, which can happen only if $d \le 4^{\kappa}$.
\end{proof}
As examples illustrating Theorem \ref{thm:featuregriesmer}, quaternary Griesmer codes $[6,3,4]$, $[10,4,6]$, $[12,6,6]$, $[16,8,8]$, and $[21,3,16]$ induce the family of EAQECCs $[[6,3\text{-}1,4;3\text{-}1]]$, $[[10,4\text{-}2,6;6\text{-}4]]$, $[[12,6\text{-}2,6;6\text{-}2]]$, $[[16,8\text{-}2,8;8\text{-}2]]$, and $[[21,3\text{-}2,16;18\text{-}17]]$, respectively, each of which contains maximal-entanglement and nonmaximal-entanglement codes that saturate the EA Griesmer bound. Here we use the notation $[[n,\kappa_{p}\text{-}\kappa_{q},d;c_{p}\text{--}c_{q}]]$, where $\kappa_p-\kappa_q = c_p-c_q>0$, to denote the family of codes $[[n,\kappa_{p},d;c_{p}]], [[n,\kappa_{p}-1,d;c_{p}-1]], \ldots,[[n,\kappa_{q}+1,d;c_{q}+1]], [[n,\kappa_{q},d;c_{q}]]$.

\begin{theorem}
The $[[n,\kappa\equiv2k-n+c,d;c]]$ EAQECCs induced by a classical $[n,k,d]_{4}$ Plotkin code will satisfy the linear EA Plotkin bound, and saturate it if and only if the EAQECC is a maximal-entanglement code or $\frac{d}{3}\frac{1-4^{-a/2}}{4^{\kappa-1}}=\frac{a}{2}$.
\label{thm:featureplotkin}
\end{theorem}
\begin{proof}
Earlier we showed that without ceiling function, Griesmer bounds (\ref{eq:classicalgriesmerbound}) and (\ref{eq:eagriesmer}) imply Plotkin bounds (\ref{eq:classicalplotkinboundqary}) and (\ref{eq:plotkinboundlinearqary}). Thus, by assumption, the classical $[n,k,d]_{4}$ Plotkin code saturates the classical Plotkin bound (\ref{eq:classicalplotkinboundqary})
\begin{align}
    n = \sum_{i=0}^{k-1} \frac{d}{4^{i}} \equiv \sum_{i=0}^{k-a/2-1}\frac{d}{4^{i}} + \sum_{i=k-a/2}^{k-1}\frac{d}{4^{i}} ,~~ a\ge2.
    \label{eq:classicalplotkinbound=}
\end{align}
We know that this classical code induces $[[n,\kappa,d;c]]$ EAQECCs. Noting that $c+\kappa\equiv n-a$, it follows that $\frac{a}{2} = n-k-c$. We can then rewrite Eq. (\ref{eq:plotkinboundlinearqary}) (setting $q=2$) for the linear EA Plotkin bound as 
\begin{align}
 n \geq \sum_{i=0}^{k-a/2-1}\frac{d}{4^{i}} + \frac{a}{2}.
 \label{eq:eaplotkincheck}
\end{align}
For $a=0$, the saturation clearly holds. Now consider the case of $a>0$. Comparing the right-hand sides of Eqs. (\ref{eq:eaplotkincheck}) and (\ref{eq:classicalplotkinbound=}), we find that saturation holds if and only if the second summand in Eq. (\ref{eq:classicalplotkinbound=}) evaluates to $\frac{a}{2}$. Noting that the summation is a geometric progression, we obtain the required result.
\end{proof}
Note that our result generalizes a result due to \cite{lu2015maximal}, which is obtained by requiring the derived EAQECC in Theorem \ref{thm:featureplotkin} to possess the same rate as the parent quaternary code, i.e., $2k-n+c {:=} k$, so that $c=n-k$, corresponding to a maximal-entanglement $[[n,k,d;n-k]]$ code.

As examples illustrating Theorem \ref{thm:featureplotkin}, consider quaternary Plotkin codes $[21,3,16]$ and $[85,4,64]$. They induce the family of EAQECCs $[[21,3\text{-}2,16;18\text{-}17]]$ and $[[85,4\text{-}3,64;81\text{-}80]]$, respectively, in which the maximal-entanglement and nonmaximal-entanglement codes saturate the EA Plotkin bound.

Paralleling a result for the Griesmer bound for CSS codes with $d\geq q$ \cite{sarvepalli2010}, we have the following results.
\begin{observation}
    Any classical $[n,k,d]_{q}$ code with $d\ge q$ satisfies
    \begin{align}
        n-k\geq d\bigg(1+\frac{1}{q}\bigg)-2.
        \label{eq:classicalgbdgeq}
    \end{align}
    \label{obs:classicalgbdgeq}
    \end{observation}
    \begin{proof}
        From Eq. (\ref{eq:classicalgriesmerbound}), since $d\ge q$ we have
        \begin{align*}
            n\geq d+\frac{d}{q}+\sum_{i=2}^{k-1}\biggl \lceil \frac{d}{q^{i}} \biggr \rceil \geq d+\frac{d}{q}+k-2,
        \end{align*}
which gives the required result.        
    \end{proof}
In Observation \ref{obs:classicalgbdgeq}, setting $q\equiv2$, we have $n-k\ge \frac{3}{2}d-2 > d-1$ $(d>q=2)$. This is consistent with the fact that the Griesmer bound is stronger than the Singleton bound. Similar to Observation \ref{obs:classicalgbdgeq}, we have the following result for EAQECCs.
\begin{observation}
    Any $[[n,k,d;c]]_{q}$ code derived from a classical $q^{2}$-ary code with $d\ge q^{2}$ satisfies
    \begin{align}
        n-k+c\geq 2d\bigg(1+\frac{1}{q^2}\bigg)-4.
        \label{eq:eagbdgeq}
    \end{align}
    \label{obs:eagbdgeq}
    \end{observation}
    \begin{proof}
        From Eq. (\ref{eq:eagriesmer}), since $d\ge q^{2}$ we have
        \begin{align*}
            \frac{n+c+k}{2}\geq d+\frac{d}{q^2}+\sum_{i=2}^{k-1}\biggl \lceil \frac{d}{q^{2i}} \biggr \rceil \geq d+\frac{d}{q^2}+k-2,
        \end{align*}
which upon simplification gives the required result.        
    \end{proof}
In Observation \ref{obs:eagbdgeq}, setting $q\equiv2$, we have $n-k+c\ge \frac{5}{2}d-4 > 2(d-1)$ $(d>q^2=4)$. This is consistent with the fact that the EA Griesmer bound is stronger than the EA Singleton bound (\ref{eq:singleton}). As a consequence of Observation \ref{obs:eagbdgeq}, we have following result for the number of correctable errors by a $q$-ary EAQECC.
\begin{observation}
    For any $[[n,k,d;c]]_{q}$ code derived from a classical $q^{2}$-ary code with $d\ge q^2$, the number of correctable errors is less than $\big\lfloor\frac{q^2(n-k+c)+2q^{2}-2}{4(q^2+1)}\big\rfloor$.
\end{observation}
\begin{proof}
     By Observation \ref{obs:eagbdgeq}, we have
     \begin{align*}
            \frac{q^2\{(n-k+c)+4\}}{2(q^2+1)}\ge d,
        \end{align*}
 which gives $\lfloor \frac{d-1}{2}\rfloor = \Big\lfloor\frac{q^2(n-k+c)+2q^{2}-2}{4(q^2+1)}\Big\rfloor$.
\end{proof}
Note that the substitutions $q^2\xrightarrow{} q :=2$ and $c\xrightarrow{}0$ yield the corresponding result for binary non-EA CSS codes \cite{sarvepalli2010}, which has been called the quantum Griesmer-Rains bound for CSS codes \cite{chandra2017}. Accordingly, the above Griesmer-based bound (\ref{eq:eagbdgeq}) may be considered as the EA version of the Griesmer-Rains bound appropriate for EAQECCs derived from a classical code.  

The analog of Theorem \ref{thm:featuregriesmer} for the EA-Griesmer-Rains bound(\ref{eq:eagbdgeq}) is as follows.
\begin{theorem}
The $[[n,\kappa \equiv 2k-n+c,d;c]]$ EAQECCs induced by a classical $[n,k,d]_{4}$ saturating the classical Griesmer-based bound (\ref{eq:classicalgbdgeq}) with $d\ge4$ will saturate the EA Griesmer-Rains bound (\ref{eq:eagbdgeq}).
\label{thm:featuregbdgeq}
\end{theorem}
\begin{proof}
By assumption, the classical $[n,k,d]_{4}$ code saturates the classical Griesmer-based bound (\ref{eq:classicalgbdgeq})
\begin{equation}
    n-k=\frac{5d}{4}-2.
    \label{eq:this}
\end{equation}
The induced \mbox{$[[n,2k-n+c,d;c]]$} code satisfies the EA Griesmer-Rains bound (\ref{eq:eagbdgeq}). Thus we have 
\begin{align}n-(2k-n+c)+c=2(n-k) \geq \frac{5}{2}d-4.
\label{eq:that}
\end{align} However, this bound must be saturated in light of Eq. (\ref{eq:this}).
\end{proof}
Note that as the $c$ drops out on the left-hand side of Eq. (\ref{eq:that}), the saturation holds irrespective of maximal and nonmaximal entanglement codes. This stands in contrast to Theorem \ref{thm:featuregriesmer}.

\section{Concatenated entanglement-assisted quantum error-correcting codes}\label{sec-eacqeccs}
As compared to a standard QECC, an EAQECC is characterized by the additional code parameter $c$, denoting the number of shared ebits, which can be optimized under different criteria  \cite{wilde2008optimal}, including the component code parameters.
In this work, for brevity of notation, we denote the concatenation of code $\mathcal{C}_o$ (outer) and $\mathcal{C}_i$ (inner) by $\mathcal{C}_o \rhd \mathcal{C}_i$. The conventional concatenation framework effectively imposes a divisibility constraint whereby the inner code rate $k_i$ must divide the outer code’s block structure for direct replacement of outer symbols by inner code blocks, or multiple outer code blocks are jointly encoded by the inner code. This leads to two types of broad procedure for the concatenation of two codes, depending on the dimension of particles (qudits) considered and the divisibility of the outer code length by the inner code rate. In the case of the GCQC framework, the divisibility constraint is circumvented because the outer symbol dimension is not matched with inner logical dimension in the same way, but instead determines the number of inner subcodes $\mathcal{Q}_j$.

These two frameworks overlap in the quantum version of the Forney scheme \cite{forney1966concatenated}, where each outer code degree of freedom corresponds to a block of the inner code, which generically leads to non-binary outer codes, i.e., the inner code is binary and the outer is a $2^{k_i}$-nary code. A number of properties of these codes in the EA paradigm are studied in Ref. \cite{fan2022entanglementassisted}. On the other hand, various authors have studied quantum code concatenation in the conventional framework using inner and outer binary codes, even with $k_i>1$, among them (see Sec. 3.4 in\cite{gaitan2008quantum} and Chap. 3.5 in \cite{gottesman1997stabilizer}). Continuing this latter line of research, here we study the concatenation of two arbitrary binary codes under the framework of conventional concatenated quantum codes.\color{black} 

Suppose we have an $[[n_o,k_o,d_o]]$ outer code $\mathcal{C}_o$ and an $[[n_i,k_i,d_i]]$ inner code $\mathcal{C}_i$. In the case of multiqubit encoding $(k_i\ge1)$, there are two construction cases (Procedures 1 and 2) for the concatenation $\mathcal{C}_o \rhd \mathcal{C}_i$. If $k_i \mid n_o$, first encode $k_o$ qubits into $n_o$ qubits using $\mathcal{C}_o$. Next each of the $\frac{n_o}{k_i}$ blocks of size $k_i$ qubits are encoded into $n_i$ qubits using $\mathcal{C}_i$. This yields an $[[\frac{n_on_i}{k_i}, k_o]]$ code with distance $d\ge \frac{d_od_i}{k_i}$ (Procedure 1 of \cite{gaitan2008quantum}). Note that Forney's procedure \cite{forney1966concatenated} can also be considered as a special realization of Procedure $1$, with $n_o \rightarrow n_ok_i$ and $k_o \rightarrow k_ok_i$. In the second case (Procedure 2), i.e., in the event that $k_i \nmid n_o$, first encode $k_ok_i$ qubits into $n_ok_i$ qubits applying $\mathcal{C}_o$ of $k_i$ blocks of length $n_o$. Next partition the resulting string into $n_o$ blocks of length $k_i$. Each of these blocks is encoded into $n_i$ qubits using $\mathcal{C}_i$. This yields a $[[n_on_i, k_ok_i]]$ code with distance $d\ge d_od_i$ (Procedure 2 of \cite{gaitan2008quantum}). The latter procedure is important for fault-tolerant quantum computation architecture \cite{gottesman2024surviving}. For example, Yamasaki and Koashi \cite{yamasaki2024time} employ this procedure for a system of nested quantum $[[2^r - 1, 2^r - 2r - 1, 3]]$ Hamming codes with $r=3,4,5,\dots$ in successive layers to achieve constant space overhead with polynomially increasing time overhead.
%Note that $k_i=1$ is considered a special case of the first construction procedure. 
As with a QECC, so with an EAQECC, different concatenation schemes arise depending on whether $k_i \mid n_o$, as discussed in the following lemmas.
\begin{lemma}
Given two EAQECCs $\mathcal{C}_o$ and $\mathcal{C}_i$  with the parameters $[[n_o,k_o,d_o;c_o]]$ and $[[n_i,k_i,d_i;c_i]]$, respectively, if $k_i \mid n_o$ then the concatenation $\mathcal{C}_o \rhd \mathcal{C}_i$ is an $[[\frac{n_on_i}{k_i}, k_o, d \geq \frac{d_od_i}{k_i};c]]$ code with $c = c_o+\frac{c_i n_o}{k_i}$.
\label{lem:divisible}
\end{lemma}
\begin{proof}
    In the case $k_i \mid n_o$, the outer encoding involves $k_o$ logical qubits into $n_o$ physical qubits and $c_o$ ebits. Next $n_o$ qubits will be partitioned into $\frac{n_o}{k_i}$ blocks with $n_i$ qubits along with $c_i$ ebits in each block. This implies the number of ebits used for the inner encoding is $\frac{c_i n_o}{k_i}$. Thus the total number of ebits used for the $[[\frac{n_on_i}{k_i}, k_o,d \geq \frac{d_od_i}{k_i},c]]$ CEAQECC is $c = c_o+\frac{c_i n_o}{k_i}$.
\end{proof}
\begin{figure}[htp]
    \centering
    \includegraphics[width=8cm]{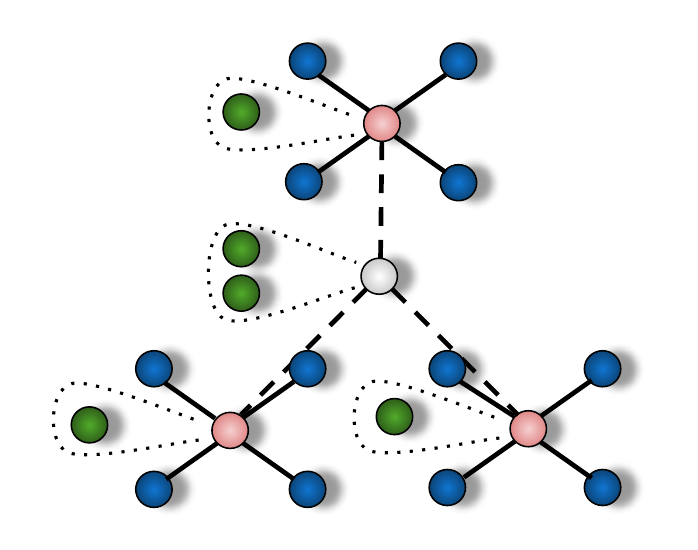}
    \caption{Graph structure of a $[[12,1;9;5]]$ CEAQECC with an outer $[[3,1,3;2]]$ code and inner $[[4,1,3;1]]$ code. The outer code encodes one logical qubit (white ball) into three physical qubits (light red balls) using two ebits (green balls). The inner code has four physical qubits (blue balls) encoded into each three of the blocks of the outer code, with one ebit (green balls) in each block.}
    \label{fig:graph-structure}
\end{figure}
\begin{lemma}
Given two EAQECCs $\mathcal{C}_o$ and $\mathcal{C}_i$,  with the parameters $[[n_o,k_o,d_o;c_o]]$ and $[[n_i,k_i,d_i;c_i]]$ respectively, if $k_i \nmid n_o$, then the concatenation $\mathcal{C}_o \rhd \mathcal{C}_i$
is a $[[n_on_i, k_ok_i, d \geq d_od_i,c]]$ code with $c=c_ok_i + c_in_o$.
\label{lem:notdivisible}
\end{lemma}
\begin{proof}
    In the case $k_i \nmid n_o$, the outer encoding involves inputting a quantum string of length $k_ok_i$ qubits. In the next step, we need to encode each block of $k_o$ qubits into $n_o$ qubits and $c_o$ ebits using code $\mathcal{C}_{o}$. This implies the number of ebits used for the outer encoding is $c_ok_i$. The inner encoding involves $k_i$ qubits in each outer code block into $n_i$ qubits and $c_i$ ebits, resulting in $c_in_o$ ebits. Thus the total number of ebits used for the $[[n_on_i, k_ok_i, d \geq d_od_i;c]]$ CEAQECC is $c=c_ok_i + c_in_o$.
\end{proof}

The above concatenation procedures can be readily generalized going to the nonbinary case, including hybrid cases, where two $q$-ary codes of mutually different alphabet sizes are concatenated. Above, we mentioned a family of non-binary code concatenation, namely that defined by Forney \cite{forney1966concatenated}, which was recast as a special case of Procedure 1.
%Another such family of $q$-ary code concatenations is studied in Refs. \cite{gottesman2024surviving, grassl2009generalized}, given by $((n_o, K))_{q_o}  \rhd ((n_i, q_o))_{q_i} \equiv ((n_on_i, K))_{q_i}$. 
%This definition can also be considered as a special realization of our Procedure $1$, as the inner code rate is 1 (in units of qudit dimension $q_o$) while the outer code length is $n_o$ qudits of dimension $q_o$.}

As an illustration of Lemmas \ref{lem:divisible} and \ref{lem:notdivisible}, consider the case of two codes $[[7,3,3;1]]$ as $\mathcal{C}_0$ and $[[6,3,3;2]]$ as $\mathcal{C}_1$, derived from the $[[8,3,3]]$ code \cite{Grassl:codetables}. 
\textcolor{black}{Let us consider the concatenation $\mathcal{C}_1 \rhd \mathcal{C}_0$, where $k_0\equiv 3 \mid n_1 \equiv 6$. By Lemma \ref{lem:divisible} above, first encode $k_1$ qubits into $n_1$ qubits and $c_1\equiv 2$ ebits using $\mathcal{C}_1$. Next each of the $\frac{n_1}{k_0} \equiv 3$ blocks of size $k_0$ qubits are encoded into $n_0$ qubits using $\frac{c_0 n_1}{k_0}=2 \equiv \beta^{\prime}$ by applying $\mathcal{C}_0$. This yields $[[\frac{n_1n_0}{k_0},k_1,d\ge \frac{d_1d_0}{k_0},c_1+\beta^{\prime}]] = [[14,3,3;4]]$.} \textcolor{black}{Note that if $k_i=1$, then $k_i \mid n_o$ quite generally. Now consider the concatenation $\mathcal{C}_0 \rhd \mathcal{C}_1$, where $k_1 \equiv 3 \nmid n_0 \equiv 7$. By Lemma \ref{lem:notdivisible}, first we encode $k_0k_1 \equiv 9$ qubits into $n_0k_1 = 21$ using $c_0k_1 = 3 \equiv \beta$ ebits by applying $\mathcal{C}_0$ on $k_1$ blocks of length $n_0$. Next partition the resulting string into $n_0$ blocks of length $k_1$. Each of these blocks is encoded into $n_1 \equiv 6$ qubits and $c_1\equiv2$ ebits using $\mathcal{C}_1$. This yields $[[n_0n_1,k_0k_1,d\ge d_0d_1,\beta+n_0c_1]] = [[42,9,9;17]]$.}
 
Here let us construct CEAQECCs with the outer code having a higher number of ebits and the inner code having a lower number of ebits and self-concatenated codes, for example $[[3,1,3;2]]$ (or $[[4,1,3;1]]$) as the outer code and  $[[5,1,3]]$ as an inner code results in $[[15, 1, 9; 2]]$ (or $[[20,1,9;1]]$) code. In a similar fashion, concatenation of the outer $[[3,1,3;2]]$ code and inner $[[4,1,3;1]]$ code results in $[[12,1,9;5]]$. 

\textcolor{black}{Graph structures of CEAQECC's are helpful in providing a pictographic representation of code concatenation and in elucidating the ebits necessitated thereby. For example, the graph structure in Fig. \ref{fig:graph-structure} depicts the concatenation $[[3,1,3;2]] \rhd [[4,1,3;1]]$ resulting in
a $[[12,1;9;5]]$ code. First, the outer layer involves one logical qubit encoded into three qubits with the aid of two ebits. This is demarcated by the dotted curve. In the inner layer, each of the three resultant qubits of the outer layer is encoded into four qubits with the aid of an ebit. The shared entanglement in this case is again depicted within the dotted curve. This allows us to directly read off the total number ebits used as 5. Here the inner code encodes single logical qubit, making it is special case of Lemma \ref{lem:divisible}. A straightforward, though more elaborate, graphic structure can be given to elucidate Lemma \ref{lem:notdivisible}.} The self-concatenation of $[[3,1,3;2]]$ and $[[4,1,3;1]]$ results in the $[[9,1,9;8]]$ and $[[16,1,9;5]]$ codes, respectively. These CEAQECCs with their metric of performances are summarized in Table \ref{table:1}. Below we discuss ways to assess the performance of CEAQECCs.

\begin{table}[t!]
\centering
\setlength{\tabcolsep}{6pt}
\renewcommand{\arraystretch}{1}
\begin{tabular}{c c c c c }
 \hline
  CEAQECCs  & $r$ & $r_e$ & $r_n$ & $\delta$  \\ [0.5ex] 
   \hline\hline
     $[[9,1,9;8]]$    & $0.1111$ & $0.8888$ & $-0.7777$ & $1$\\
     $[[12,1,9;5]]$   & $0.0833$ & $0.4166$ & $-0.3333$  & $0.75$ \\
  $[[15,1,9;2]]$   & $0.0666$ & $0.1333$ & $-0.6666$  & $0.6$ \\
   $[[16,1,9;5]]$    & $0.0625$ & $0.3125$ & $-0.25$  & $0.5625$\\
  $[[20,1,9;1]]$   & $0.05$ & $0.05$ & $0$  & $0.45$ \\[1ex]
 \hline
\end{tabular}
\caption{CEAQECCs with code rates, entanglement-assistance rates, net rates, and relative distances.} 
\label{table:1}
\end{table}

\section{Order dependence of the performance of CEAQECCs}\label{sec-optimizing_order_less_ebits}
Among the parameters that characterize the performance of CEAQECCs, three parameters are the three following parameters: the number of ebits $c$, the logical error probability $p_L$, and its pseudothreshold $p_{\theta}$; the latter two are applicable also to general QECCs. Ideally, we want low values of $c$ and $p_L$ and a high value of $p_{\theta}$. Moreover, in the context of EAQECCs, there is the added burden that the shared ebits are noiseless \cite{fan2022entanglementassisted}, which in turn requires entanglement distillation protocols \cite{bennett1996purification, bennett1996mixed}. Thus, it is of interest to study how the order of concatenation can be optimized \cite{chamberland2017error, dash2023concatenating} in order to minimize the number of ebits required to perform efficient error correction. 

\begin{theorem}
Consider the concatenation of an \mbox{$[[n-c_1,k;c_1]]$} EAQECC $\mathcal{C}_1$ and an $[[n-c_2,k;c_2]]$ EAQECC $\mathcal{C}_2$, where $0 \leq c_1 < c_2$ and both are derived from a standard  $[[n,k]]$ stabilizer code. If either (a) $k\mid n-c_1$ and $k\mid n-c_2$ or (b) $k\nmid n-c_1$ and $k\nmid n-c_2$, then the concatenated code \mbox{$\mathcal{C}_2 \rhd \mathcal{C}_1$} requires fewer ebits than the concatenated code $\mathcal{C}_1 \rhd \mathcal{C}_2$.  
    \label{thm:lessebits2}
\end{theorem}
\begin{proof}
Let $c$ ($c^{\prime}$) denote the ebits of the concatenated code $\mathcal{C}_1 \rhd \mathcal{C}_2$ ($\mathcal{C}_2 \rhd \mathcal{C}_1$).
\begin{description}
    \item[$k \mid n-c_2, k \mid n-c_1$] By Lemma \ref{lem:divisible} we have $c=c_1+\frac{c_2(n-c_1)}{k}$ and $c^{\prime} = c_2+\frac{c_1(n-c_2)}{k}$. It follows that $c-c^{\prime} = (\frac{n}{k}-1)(c_2-c_1) > 0$.
    \item[$k \nmid n-c_2, k \nmid n-c_1$] By Lemma \ref{lem:notdivisible} we have $c=c_1k+c_2n - c_1c_2$ and $c^{\prime} = c_2k+c_1n - c_1c_2$. It follows that $c-c^{\prime} = (n-k)(c_2-c_1) > 0$.
\end{description}
Note that the ebit difference is the same in both cases apart from a normalization with respect to $k$.
\end{proof}
For the other two cases, i.e., $k_1 \nmid n_2$ and $k_2 \mid n_1$, and $k_1 \mid n_2$ and $k_2 \nmid n_1$, the code \mbox{$\mathcal{C}_2 \rhd \mathcal{C}_1$} may not necessarily use fewer ebits than the concatenated code $\mathcal{C}_1 \rhd \mathcal{C}_2$. The following examples illustrate these cases, respectively. Consider $[[4,2,2;1]]$ as $\mathcal{C}_1$ and $[[3,2,2;2]]$ as $\mathcal{C}_2$, both derived from $[[5,2,2;0]]$ \cite{Grassl:codetables}. Here the numbers of ebits required by the concatenated codes $\mathcal{C}_1 \rhd \mathcal{C}_2$ and $\mathcal{C}_2 \rhd \mathcal{C}_1$ are $c=5$ and $c^{\prime}=7$, respectively. For the last case, consider $[[9,2,4;1]]$ as $\mathcal{C}_1$ and $[[8,2,4;2]]$ as $\mathcal{C}_2$, both derived from $[[10,2,4;0]]$ \cite{Grassl:codetables}. Here the numbers of ebits used by the concatenated code $\mathcal{C}_1 \rhd \mathcal{C}_2$ and $\mathcal{C}_2 \rhd \mathcal{C}_1$ are $c=20$ and $c^{\prime}=6$, respectively.

The above theorem indicates a family of codes,  where the number of ebits required under the concatenation of a pair of members is in general order dependent. By contrast, the following theorems present a family where the number of ebits required is order independent.

\begin{theorem}
    Given two EAQECCs $[[n,k;c]]$ and \mbox{$[[n+m,k+m;c]]$} ($m \in \mathbb{Z}_{\ge 0}$), the number of ebits required under the concatenation with $k_{\rm inner} \nmid n_{\rm outer}$ is independent of ordering.
      \label{thm:switch_invariant_same_c}
\end{theorem}
\begin{proof}
The number of ebits required under concatenation with either ordering is given by $c(k+m+n)$.
\end{proof}
As an example concatenation of EAQECCs $[[3,2,2;1]]$ \cite{Grassl:codetables} and $[[5,4,2;1]]$ \cite{luo2023constructing} by both ordering give a $[[15,8,4;7]]$ CEAQECC. Below we discuss a similar result. They have the following property.
\begin{theorem}
    With the existence of an $[[n,k,d;c]]$ code, extended codes $[[n+1,k;c+1]]$ and $[[n,k-1;c+1]]$ also exist \cite{lai2013dualities}. The number of ebits required under the concatenation of the extended codes is independent of ordering, provided $k_{\rm inner} \nmid n_{\rm outer}$.
      \label{thm:switch_invariant_same_c+1}
\end{theorem}
\begin{proof}
The number of ebits required under concatenation with either ordering is given by $(n+k)(c+1)$.
\end{proof}
The following result is for the family of EA repetition codes.
\begin{theorem}
    For the family of $[[n,1;n-1]]$ EAQECCs ($n \in \mathbb{Z}^+$ and $Z\ge3$), the number of ebits required under the concatenation of any pair of members is independent of ordering.
      \label{thm:switch_invariant}
\end{theorem}
\begin{proof}
Consider two codes in this family, given by EAQECCs $[[n,1;n-1]]$ and \mbox{$[[m,1;m-1]]$} where $m \neq n$. The number of ebits required for concatenating the former with the latter is given by
$(n-1)\times 1 +(m-1)\times n = nm-1$, which is invariant when $n$ and $m$ are interchanged.
\end{proof}

Note that both $[[n,1,n;n-1]]$ \cite{lai2013entanglement} and $[[n,1,n-1;n-1]]$ \cite{lai2013dualities} EA repetition codes for odd and even $n$, respectively, are covered under Theorem \ref{thm:switch_invariant}. The term ``repetition'' here indicates that these two families of codes are constructed using the parity check matrix of classical $[n,1,n]$ and $[n-1,1,n-1]$ repetition codes, respectively \cite{lai2013entanglement, lai2013dualities}. 

Now we construct different sets of CEAQECCs by concatenating the $[[5,1,3]]$ or $[[3,1,3;2]]$ repetition code ($[[3,1,3;2]]_{R}$) with $[[4,1,3;1]]$ \cite{brun2006correcting}, where the latter two codes are not derivatives of the $[[5,1,3]]$ code. We check how the order of concatenation affects the number of ebits, logical error probability, and pseudothreshold for different sets of CEAQECCs.

Consider the probability of an error on single qubit is $p$. The logical error probability of the $[[5,1,3]]$ code is
\begin{align}
    p_L^{[[5,1,3]]}=1 - (1 - p)^5 - 5 (1 - p)^4 p.
    \label{eq: logicalrate_five}
\end{align}
The stabilizer for the $[[3,1,3;2]]_{R}$ code \cite{lai2013entanglement} is
\begin{align}
    \mathcal{S}^{[[3,1,3;2]]_{R}}&=
    \begin{matrix}
 Z & Z & I  & \vline & I & Z \\
 I & X & X  & \vline & I & X \\
 I & Z & Z  & \vline & Z & I \\
 X & X & I  & \vline & X & I  \\
 \label{eq: stabilizers_3_1_3_2_rep}
\end{matrix},
\end{align}
where the vertical bar indicates a partition between Alice and Bob's possession of qubits. \textcolor{black}{As this code is constructed with the assumption that shared ebits of Bob's part are noiseless, the six errors $X_4$, $Y_4$, $Z_4$, $X_5$, $Y_5$, and $Z_5$ are absent by design. Like the [[5,1,3]] code, which is perfect (in that it saturates the Hamming bound), the present [[3,1,3;2]]$_R$ code is also perfect [in that it satisfies the EA Hamming bound (\ref{eq:eahb})]. Therefore, the syndrome space freed up by the noiseless assumption can be exploited to correct six weight-2 errors on the first three qubits.} Specifically, the set of errors correctable by this code is \cite{lai2012entanglement-assisted}
\begin{align}
    \mathcal{E}^{[[3,1,3;2]]_{R}}&= \langle I,X^1,X^2,X^3,Z^1,Z^2,Z^3,Y^1,Y^2,Y^3,\nonumber\\
    &X^1Z^2,X^1Z^3,Z^1X^2,X^2Z^3, Z^1X^3,Z^2X^3 \rangle.
    \label{eq:correct3ea}
\end{align}
\textcolor{black}{Each of the above errors produces a distinct syndrome of the $[[3,1,3;2]]_R$ code and can thus be corrected by Bob provided his qubits remain noiseless.
The logical error probability then follows as the probability that none of the errors in the list Eq. (\ref{eq:correct3ea}) occurs:}
\begin{align}
    p_L^{[[3,1,3;2]]_{R}}=1 - (1 - p)^3 - 3 (1 - p)^2 p - \frac{2}{9} (1 - p) p^2.
    \label{eq: logicalrate_rep}
\end{align}
The $[[4,1,3;1]]$ code is constructed from the classical $[4,2,3]$ quaternary code, and it can correct arbitrary single-qubit errors in the first four qubits \cite{brun2006correcting}. Thus, the logical error probability of the $[[4,1,3;1]]$ code is
\begin{align}
    p_L^{[[4,1,3;1]]}=1 - (1 - p)^4 - 4 (1 - p)^3 p.
    \label{eq: logicalrate_four}
\end{align}

If $\mathcal{C}_{a}$ and $\mathcal{C}_{b}$ are two quantum error-correcting codes with logical error probabilities $p_L^{\mathcal{C}_{a}}(p)$ and $p_L^{\mathcal{C}_{b}}(p)$, respectively, then the logical error probability of the concatenated code $\mathcal{C}_{a} \rhd \mathcal{C}_{b}$  is expressed as \cite{gaitan2008quantum, dash2023concatenating}
\begin{align}
  p_L^{\mathcal{C}_{a} \rhd \mathcal{C}_{b}}(p) = p_L^{\mathcal{C}_{a}}[p_L^{\mathcal{C}_{b}}(p)].
  \label{eq: cc_logicalrate}
\end{align}
By way of example, consider concatenating a $[[5,1,3]]$ code with a $[[3,1,3;2]]_{R}$ or $[[4,1,3;1]]$ code. Here the subscript $R$ denotes the repetition code, in contrast to the Bowen code, derived from the $[[5,1,3]]$ code, indicated by $[[3,1,3;2]]_{B}$. 

With a slight abuse of notation, we have concatenated codes $[[5,1,3]] \rhd [[3,1,3;2]]_{R} \equiv [[15,1,9;10]]_{(R)}$ and $[[3,1,3;2]]_{R} \rhd [[5,1,3]] \equiv [[15,1,9;2]]_{(R)}$.  The logical error probabilities for these two CEAQECCs are evaluated using Eq. (\ref{eq: cc_logicalrate}) and are depicted in Fig. \ref{fig:ler15102}. Here notice that $p_L$ is dependent on the concatenation ordering. \textcolor{black}{We recollect that the error pseudothreshold $p_{\theta}$ is defined as the maximum single-qubit error $p$ such that encoding remains effective, i.e., $\max_p p_L(p)\le p$. It may be noted that $p_{\theta}$ is insensitive to recursive concatenation.}
A similar dependence is borne out in the pseudo-threshold $p_{\theta}$, which for $[[15,1,9;10]]_{(R)}$ and  $[[15,1,9;2]]_{(R)}$ is found to be $0.2441$ and $0.2284$, respectively. Interestingly, no such dependence is found in either $p_L$ or $p_{\theta}$ when $[[3,1,3;2]]_R$ above is replaced by $[[3,1,3;2]]_{B}$, showing that the parameters $n$,$k$,$d$, and $c$ alone do not determine the properties of the concatenated code. As another example, for $[[5,1,3]] \rhd [[4,1,3;1]] \equiv [[20,1,9;5]]$ and $[[4,1,3;1]] \rhd [[5,1,3]] \equiv [[20,1,9;1]]$ CEAQECCs, $p_{\theta}$ is found to be $0.1877$ and $0.1622$, respectively.

%\begin{widetext}
\begin{figure}[t]
         \centering
         \includegraphics[width=8cm]{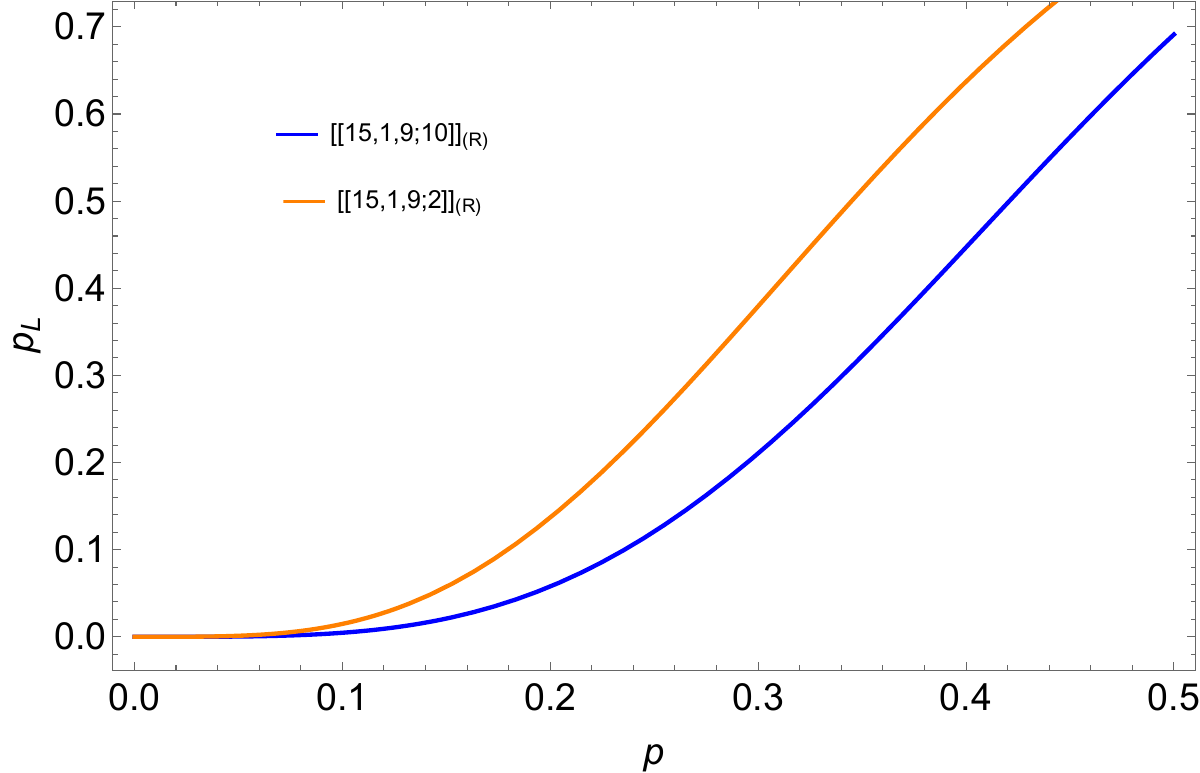}
         \caption{Logical error probabilities for the concatenation of $[[5,1,3]]$ and $[[3,1,3;2]]$ repetition codes.}
         \label{fig:ler15102}
    \end{figure}
%\end{widetext}

\section{Saturation or violation of various bounds by specific CEAQECC classes}\label{sec-violation_or_saturation_ceaqeccs}

From the concatenation of QECCs, we know that if a standard stabilizer code saturates the quantum Singleton bound, this does not imply that its self-concatenated code also saturates that bound. A simple illustration is that the $[[5,1,3]]$ code saturates the quantum Singleton bound while its twofold self-concatenated version, a $[[25, 1, 9]]$ code, does not saturate it. However, certain families of EA MDS codes can be closed under concatenation. One such case is proved below.
 \begin{theorem}
    Any $l$-fold concatenation of codes taken from the family of the non-degenerate $[[n,1,n;n-1]]$ EA repetition codes for $n$ odd saturates the trio of EA Singleton, EA Griesmer and linear EA Plotkin bounds, where $l \in \mathbb N_{> 0}$. 
    \label{th: l-fold}
\end{theorem}
\begin{proof}
    The concatenation $[[n,1,n;n-1]] \rhd [[m,1,m;m-1]]$ yields a $[[mn,1,mn;mn-1]]$ CEAQECC belonging to the same family. 
    By Eq. (\ref{eq:singleton}) we find that the $[[mn,1,mn;mn-1]]$ CEAQECC saturates the EA Singleton bound. Noting that $k=1$, by Theorem \ref{thm:3in1} it follows that this CEAQECC also saturates the EA-Griesmer and linear EA Plotkin bounds. The required result is then obtained by recursive concatenation.
\end{proof}
In contrast to the above, any $l$-fold concatenation of codes taken from the family of $[[n,1,n-1;n-1]]$ EA repetition codes for even $n \ge 4$ does not saturate the above-mentioned bounds. 
 
The above self-concatenation of the $[[5,1,3]]$ code shows that component MDS codes need not result in an MDS concatenated code. By contrast, here is an example where concatenation preserves the MDS property. The MDS CEAQECCs $[[5n, 1, 3n; n -1]]$ and $[[4n, 1, 3n; 2n -1]]$ are constructed by choosing an outer MDS $[[n, 1, n; n -1]]$ code, where $n\geq 3$ and is odd, and the inner code given by MDS $[[5, 1, 3]]$ and $[[4, 1, 3; 1]]$ codes. By Theorem \ref{thm:3in1}, these codes also saturate the EA Griesmer and linear EA Plotkin bounds. 
%\textcolor{green}{\textbf{These codes are reported in \cite{fan2022entanglementassisted} for the violation of the EA-Hamming bound.}} 
Note that if we relax the requirement that $n$ is odd in Theorem \ref{th: l-fold}, then the EAQECC can be shown not to exist \cite{lai2013dualities}. 

In principle, degenerate quantum codes may violate the quantum Hamming bound (QHB), and the QHB violation implies degeneracy. Four families of degenerate EAHB-violating CEAQECCs obtained from component EAHB-satisfying codes are given in \cite{fan2022entanglementassisted}. In general, the EAHB-violating property of the component codes by itself does not determine EAHB violation by the concatenated code. In the following result, we construct further degenerate EAHB-violating CEAQECCs, each from a pair of the component codes, where precisely one of the pair violates the EAHB. The proof is analogous to that for the CEAQECCs in \cite{fan2022entanglementassisted} and EAQECCs in \cite{li2014entanglement} that violate the EAHB.

\begin{theorem}
Consider the EAHB-violating (degenerate) $[[8, 1, 5; 1]]$ code $\mathcal{C}_{1}$ concatenated with the EAHB-satisfying $[[n, 1, n; n-1]]$ (odd $n$) code $\mathcal{C}_{a}$ or $[[n, 1, n-1; n-1]]$ (even $n$) code $\mathcal{C}_{b}$. The family of CEAQECCs $\mathcal{C}_a \rhd \mathcal{C}_1$ (odd $n\geq 3$) and $\mathcal{C}_b \rhd \mathcal{C}_1$ (even $n \geq  10$) violates the EAHB, whereas the family of CEAQECCs $\mathcal{C}_1 \rhd \mathcal{C}_a$ and $\mathcal{C}_1 \rhd \mathcal{C}_b$ satisfies the EAHB. 
 \label{th: qhbviolationth4}
 \end{theorem}
 \begin{proof}
     For $n \geq 2$ and $1<i<n$ we have the binomial inequality \cite{macwilliams1978theory,fan2022entanglementassisted}:\\
     \begin{align}
         \frac{2^{nH_2(\beta)}}{\sqrt{8n\beta(1-\beta)}} \leq \binom{n}{i} \leq \frac{2^{nH_2(\beta)}}{\sqrt{2 \pi n\beta(1-\beta)}},
         \label{eq:stirling}
     \end{align}
     where $H_2(x)=-x \log_2x-(1-x) \log_2(1-x)$ is the binary entropy function and $\beta=\frac{i}{n}$. Consider the family of CEAQECCs $[[8n, 1, 5n; 2n-1]]$. Let $n=2m+1$ $(m\geq 1)$. From the above binomial inequality we have
     \begin{subequations}
        \begin{align}
         \sum_{i=0}^{\lfloor(5n-1)/2\rfloor} 3^{i} \binom{8n}{i} & \geq 3^{5m+2}  \binom{16m+8}{5m+2}\label{eq:qhbviolationproofa}\\
         & \geq  \frac{3^{5m+2} 2^{(16m+8)H_2(\beta)}}{\sqrt{8(16m+8)\beta(1-\beta)}} \label{eq:qhbviolationproofb}\\ &\geq 2^{20m+8}=2^{10n-2} \forall m \geq 1,
         \label{eq:qhbviolationproofc}
     \end{align}       
     \label{eq:qhbviolationproof}
     \end{subequations}
     where $\beta=\frac{5m+2}{16m+8}$. Here the right-hand side of Eq. (\ref{eq:qhbviolationproofa}) comes from the last term in the expansion of the left-hand side We obtain the right-hand side in Eq. (\ref{eq:qhbviolationproofb}) by the Stirling-like identity (\ref{eq:stirling}). One readily determines that for sufficiently large $m$ the right-hand side of Eq. (\ref{eq:qhbviolationproofb}) exceeds that of Eq. (\ref{eq:qhbviolationproofc}). In fact, we find that the CEAQECCs $[[8n, 1, 5n; 2n-1]]$ violate the EAHB for odd $n\geq 3$. The violation of EAHB by the family of CEAQECCs $\mathcal{C}_b \rhd \mathcal{C}_1$ ($n\geq 10$) follows in an analogous manner. By contrast, the EAHB is satisfied by the family of CEAQECCs $\mathcal{C}_1 \rhd \mathcal{C}_{a/b}$.

To show that consider the family of CEAQECCs $\mathcal{C}_1 \rhd \mathcal{C}_{a} \equiv [[8n, 1, 5n; 8n-7]]$. As in the above procedure, we find
         \begin{align}
         \sum_{i=0}^{\lfloor(5n-1)/2\rfloor} 3^{i} \binom{8n}{i} &\leq 2^{32m+8}=2^{16n-8} \forall m \geq 1,
         \label{eq:qhbviolationproof2}
     \end{align}
implying satisfaction of the EAHB for odd $n\geq 3$.
\end{proof}

We list below other examples of EAHB-violating CEAQECCs. Consider codes $\mathcal{C}_2 \equiv [[7,1,5;2]]$ and $\mathcal{C}_4 \equiv [[9,1,7;4]]$ \cite{lai2013entanglement}. Then the following CEAQECCs can similarly be shown to violate the EAHB:
\begin{center}
\setlength{\tabcolsep}{25pt}
\begin{tabular}{c  c}
\hline\hline
Code & Condition on $n$\\
\hline
$\mathcal{C}_a \rhd \mathcal{C}_2$ & odd $n \geq 3$ \\
$\mathcal{C}_a \rhd \mathcal{C}_4$ & odd $n \geq 11$ \\ 
$\mathcal{C}_b \rhd \mathcal{C}_2$ &even $n \geq 18$ \\ 
$\mathcal{C}_b \rhd \mathcal{C}_4$ & even $n \geq 52$ \\
\hline\hline
\end{tabular}
\end{center}

With the existence of an $[[n,k,d;c]]$ code, a extended code $[[n+1,k,d;c+1]]$ also exists \cite{lai2013dualities}. This can be used to search for other EAHB-violating CEAQECCs. As one example, codes $\mathcal{C}_a$ and  $\mathcal{C}_b$ of Theorem \ref{th: qhbviolationth4} can be transformed to new codes $\mathcal{C}_{a^{'}}\equiv[[n+1, 1, n; n]]$ (odd $n \geq 3$)  and $\mathcal{C}_{b^{'}}\equiv [[n+1, 1, n-1; n]]$ (even $n \geq 4$), respectively. It may be checked using the method of Theorem \ref{th: qhbviolationth4} that the family of CEAQECCs $\mathcal{C}_{a^{'}} \rhd \mathcal{C}_1$ (odd $n\geq 9$) and $\mathcal{C}_{b^{'}} \rhd \mathcal{C}_1$ (even $n \geq  16$) also violate the EAHB. Note that other methods of code modification can be used \cite{luo2023constructing} to obtain new codes in the search for the EAHB violation. 

As an illustration of the order-dependent violation of the EAHB by the resultant of concatenating two component codes, consider CEAQECCs $\mathcal{C}_a \rhd \mathcal{C}_1$ and $\mathcal{C}_1 \rhd \mathcal{C}_a$. Figure \ref{fig:he} depicts the Hamming efficiency $\varphi$ for the two concatenated families over a range of $n$ values, showing that $\varphi>1$ for the former and $\varphi<1$ for the latter. 
\begin{figure}[htp]
    \centering
    \includegraphics[width=8cm]{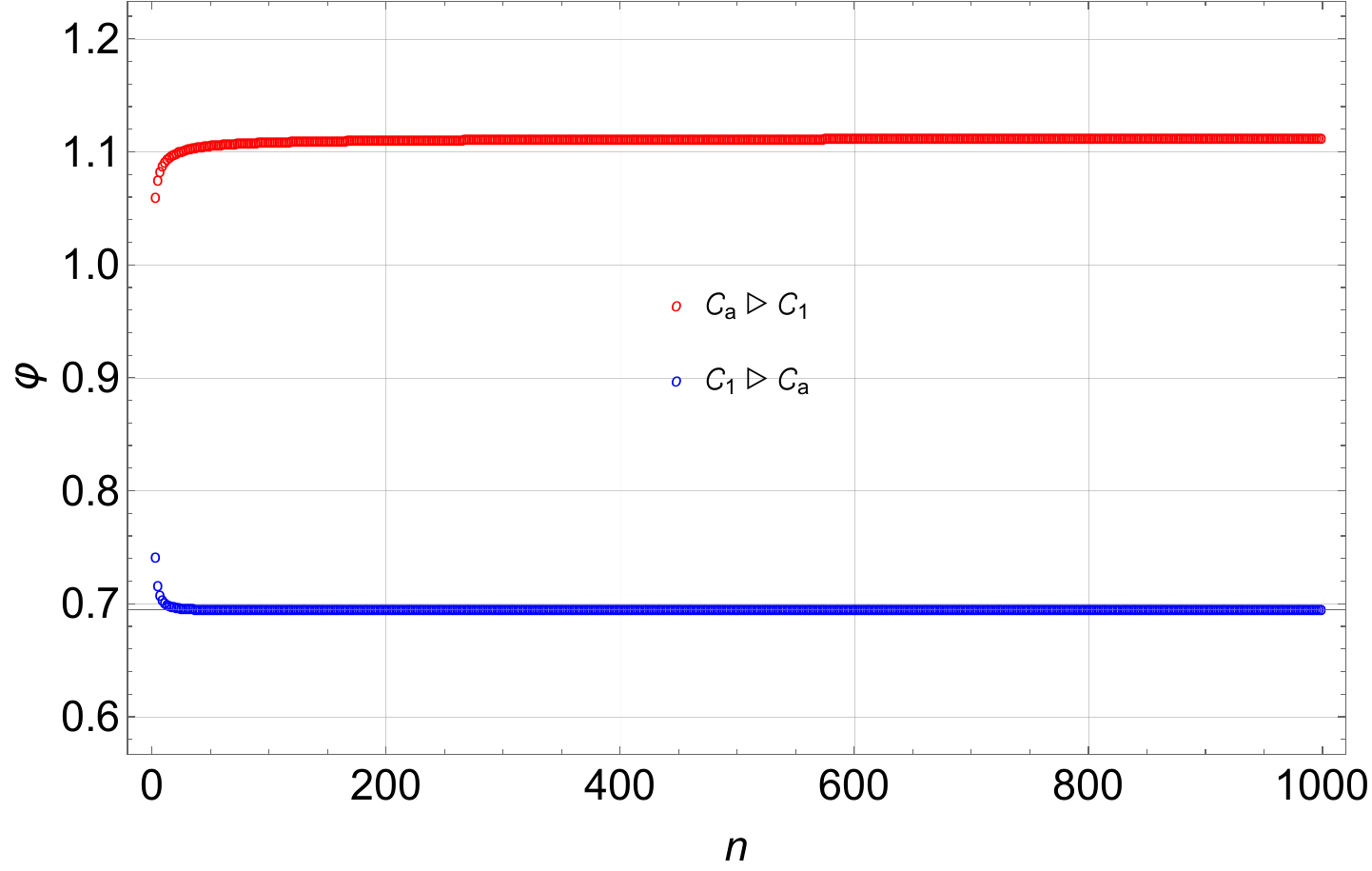}
    \caption{Hamming efficiency plot for the two families of CEAQECCs $\mathcal{C}_{a} \rhd \mathcal{C}_{1}$ and $\mathcal{C}_{1} \rhd \mathcal{C}_{a}$, where $\mathcal{C}_{1} \equiv [[8,1,5;1]]$ and $\mathcal{C}_{a} \equiv [[n, 1, n; n-1]]$ (odd $n$). The former concatenated code violates the EA Hamming bound, while the latter concatenated code does not.}
    \label{fig:he}
\end{figure}

\section{Conclusion}\label{sec-conclusion}

In this work we have investigated the scope and properties of various bounds on EAQECCs and their concatenation. We have derived the general expression of shared entanglement for concatenated codes, encompassing both the case where the outer code length is divisible by the inner code rate and the case where it is not. We show that the concatenating order plays a role in determining in general the properties of the resultant codes, such as the shared entanglement or error-probability threshold. Even so, we have constructed families of pairs of EAQECCs such that concatenating the partner codes within a pair results in CEAQECCs with order-independent ebits. While it is known that maximal-entanglement EAQECCs obtained from a classical quaternary Griesmer or Plotkin code saturate the EA Griesmer or linear EA Plotkin bound, respectively, here we have shown that nonmaximal-entanglement EAQECCs also can saturate these bounds, provided a corresponding condition on code distance is met. We have illustrated these results by presenting several families of such nonmaximal-entanglement EAQECCs. We have also computed an upper bound on the number of correctable errors of EAQECCs satisfying $d \ge q^2$, which may be considered as the EA extension of the Griesmer-Rains bound. We have further constructed families of pairs of EAQECCs such that concatenating the partner codes within a pair results in CEAQECCs that violate the EAHB in an order-dependent way.

Finally, we indicate a number of open questions suggested by our work. In \cite{grassl2021entanglementassisted}, codes of the type $[[n,k,d;k]]$ such that $d\ge\frac{n}{2}+1$ were shown to violate the EA Singleton bound (\ref{eq:singleton}). As the EA Griesmer bound is a tighter bound, these counterexamples will also violate this bound. However, these codes satisfy the higher-distance EA Singleton bound (\ref{eq:qsb_d_geq}), derived in \cite{grassl2022entropic}. Thus an interesting question here is what, if any, would be a corresponding refinement of the EA Griesmer bound as applicable to small or large code distance, with the subtlety that this distance relation can depend on the alphabet size $q$.

Furthermore, in Theorems \ref{thm:featuregriesmer} and \ref{thm:featureplotkin} the saturation of the EA Griesmer or linear EA Plotkin bound is related to a constraint on the alphabet size $q$. Specifically, in the former case, we obtained the condition $d \le 4^{\kappa}$, or more generally $d \le q^{2\kappa}$, while in the latter case, we obtained $2d(1-2^{-a})=3a4^{\kappa-1}$, or more generally $2d(1-q^{-a})=a(q^2-1)q^{2(\kappa-1)}$. For standard CSS codes it has been proved that for $q\ge 5$ they must satisfy the quantum Hamming bound, and there is a similar condition on alphabet size for non-CSS quantum codes that satisfy the quantum Hamming bound \cite{sarvepalli2010}. Thus our second question concerns the condition on alphabet size $q$ for an $[[n,k,d;c]]_{q}$ EAQECC to satisfy the EAHB. In particular, this condition should explain the order-dependent violation of the EAHB demonstrated in Theorem \ref{th: qhbviolationth4}.

Finally, we recollect that EA subsystem codes combine entanglement assistance and passive error correction \cite{hsieh2007general}. Degenerate EAQECCs and CEAQECCs can violate the EAHB, as shown above in Sec. \ref{sec-violation_or_saturation_ceaqeccs} and in Refs. \cite{li2014entanglement, fan2022entanglementassisted}. Further, it was shown in Ref. \cite{klappenecker2007subsystem} that some quantum subsystem codes exist that violate the quantum subsystem Hamming bound. Thus, a third question is whether the EA subsystem Hamming bound can be violated by EA subsystem codes. For the above mentioned reasons, it is reasonable to expect that the answer to this question is in the affirmative.

\acknowledgments
The authors thank M. Grassl for helpful comments.
N.R.D. and S.D. acknowledge financial support from the Department of Science and Technology, Ministry of Science and Technology, India, through the INSPIRE fellowship and the University Grants Commission India through the NET fellowship, respectively. R.S. acknowledges partial support from the Science and Engineering Research Board through Grant No. CRG/2022/008345. N.R.D. is sincerely grateful to the Poornaprajna Institute of Scientific Research for its hospitality and conducive environment during his academic visit.

\appendix

\bibliography{reference}
\end{document}